\documentclass[a4paper, 10pt]{article}

\usepackage{geometry}
\usepackage{authblk}
\usepackage{amsthm}
\usepackage[dvipdfmx]{graphicx}
\usepackage{comment}
\usepackage{amssymb}
\usepackage{bm}
\usepackage{ascmac}
\usepackage{algorithm,algpseudocode,float}
\usepackage{color}
\usepackage{mypackage}
\usepackage{breqn}
\usepackage{xspace}

\makeatletter
\newenvironment{breakablealgorithm}
   {%
   \begin{center}
    \refstepcounter{algorithm}%
    \hrule height.8pt depth0pt \kern2pt%
    \renewcommand{\caption}[2][\relax]{%
     {\raggedright\textbf{\ALG@name~\thealgorithm} ##2\par}%
     \ifx\relax##1\relax %
      \addcontentsline{loa}{algorithm}{\protect\numberline{\thealgorithm}##2}%
     \else %
      \addcontentsline{loa}{algorithm}{\protect\numberline{\thealgorithm}##1}%
     \fi
     \kern2pt\hrule\kern2pt
    }
   }{%
    \kern2pt\hrule\relax%
   \end{center}
   }
\makeatother

\newtheorem{theorem}{Theorem}[section]
\newtheorem{lemma}{Lemma}[section]

\title{\textbf{Gathering with a strong team in weakly Byzantine environments}}

\author[1]{Jion Hirose}
\author[2]{Junya Nakamura}
\author[1]{Fukuhito Ooshita}
\author[1]{Michiko Inoue}

\affil[1]{Nara Institute of Science and Technology}
\affil[2]{Toyohashi University of Technology}

\date{July 31, 2020}

\begin{document}

\maketitle

\begin{abstract}
We study the gathering problem requiring a team of mobile agents to gather at a single node in arbitrary networks. The team consists of $k$ agents with unique identifiers (IDs), and $f$ of them are weakly Byzantine agents, which behave arbitrarily except falsifying their identifiers. The agents move in synchronous rounds and cannot leave any information on nodes. If the number of nodes $n$ is given to agents, the existing fastest algorithm tolerates any number of weakly Byzantine agents and achieves gathering with simultaneous termination in $O(n^4\cdot|\Lambda_{good}|\cdot X(n))$ rounds, where $|\Lambda_{good}|$ is the length of the maximum ID of non-Byzantine agents and $X(n)$ is the number of rounds required to explore any network composed of $n$ nodes. In this paper, we ask the question of whether we can reduce the time complexity if we have a strong team, i.e., a team with a few Byzantine agents, because not so many agents are subject to faults in practice. We give a positive answer to this question by proposing two algorithms in the case where at least $4f^2+9f+4$ agents exist. Both the algorithms take the upper bound $N$ of $n$ as input. The first algorithm achieves gathering with non-simultaneous termination in $O((f+|\Lambda_{good}|)\cdot X(N))$ rounds. The second algorithm achieves gathering with simultaneous termination in $O((f+|\Lambda_{all}|)\cdot X(N))$ rounds, where $|\Lambda_{all}|$ is the length of the maximum ID of all agents. The second algorithm significantly reduces the time complexity compared to the existing one if $n$ is given to agents and $|\Lambda_{all}|=O(|\Lambda_{good}|)$ holds.
\end{abstract}

\section{Introduction}
\subsection{Background}
Mobile agents (in short, agents) are software programs that move autonomously and perform various tasks in a distributed system.
A task that collects multiple agents on the same node is called a \emph{gathering}, and this task has been widely studied from the theoretical aspect of distributed systems~\cite{Pelc2019book}. 
By accomplishing this task, the agents can exchange information with each other more efficiently, and it becomes easy to carry out future cooperative behaviors.

In operations of large-scale distributed systems, we cannot avoid facing faults of agents.
Among them, Byzantine faults are known to be the worst faults because Byzantine faults do not make any assumption about the behavior of faulty agents (called \emph{Byzantine agents}).
For example, Byzantine agents can stop and move at any time apart from their algorithm, and tell arbitrary wrong information to other agents.

In this study, we consider the deterministic gathering problem with Byzantine agents and propose two synchronous gathering algorithms for the problem.

\subsection{Related works}
\label{sec:related-work}
\begin{table*}
  \caption{A summary of synchronous Byzantine gathering algorithms with unique IDs.
  Here, $n$ is the number of nodes, $N$ is the upper bound of $n$, $|\lambda_{good}|$ is the length of the smallest ID among non-Byzantine agents, $|\Lambda_{good}|$ is the length of the largest ID among non-Byzantine agents, $|\Lambda_{all}|$ is the length of the largest ID among agents, $k$ is the number of agents, and $f$ is the number of Byzantine agents.}
  \label{tab_related_works}
  \scalebox{0.8}{
  \begin{tabular}{ccccccc}
    \hline
    &Input&Byzantine&\begin{tabular}{c}Condition of \\\#Byzantine agents\end{tabular}&\begin{tabular}{c}Simultaneous\\termination\end{tabular}&Time complexity\\
    \hline \hline
    \cite{Dieudonne2014} & $n$ & Weak & $f+1\leq k$ & Possible & $O(n^4\cdot |\Lambda_{good}|\cdot X(n))$\\
    \cite{Dieudonne2014} & $f$ & Weak & $2f+2\leq k$ & Possible & Poly.~of $n$ \& $|\Lambda_{good}|$\\
    \cite{Bouchard2016} & $n,f$ & Strong & $2f+1\leq k$ & Possible & Exp.~of $n$ \& $|\Lambda_{good}|$\\
    \cite{Bouchard2016} & $f$ & Strong & $2f+2\leq k$ & Possible & Exp.~of $n$ \& $|\Lambda_{good}|$\\
    \cite{Bouchard2018} & $\lceil \log\log n \rceil$ & Strong & $5f^2+7f+2\leq k$ & Possible & Poly.~of $n$ \& $|\lambda_{good}|$\\ 
    \begin{tabular}{c}Proposed algorithm 1\end{tabular} & $N$ & Weak & $4f^2+9f+4\leq k$ & Impossible & $O((f+|\Lambda_{good}|)\cdot X(N))$\\
    \begin{tabular}{c}Proposed algorithm 2\end{tabular} & $N$ & Weak & $4f^2+9f+4\leq k$ & Possible & $O((f+|\Lambda _{all}|)\cdot X(N))$\\
    \hline
  \end{tabular}
  }
\end{table*}
The gathering problem has been studied for the first time by Schelling~\cite{Schelling1960}.
In that paper, the author studied the gathering problem of exactly two agents, called the rendezvous problem.
After that, the rendezvous problem and its generalization, the gathering problem, have been widely studied in various environments that combine agent synchronization, anonymity, presence/absence of memory on a node (called whiteboard), presence/absence of randomization, and topology, etc.~\cite{Pelc2019book}.
The purpose of these studies is to clarify the solvability of the gathering problem and its costs (e.g., time, the number of moves, and memory space, etc.) if solvable.
The rest of this section describes the deterministic gathering problem in arbitrary networks, on which we focus in this paper.

Many of the papers dealing with the rendezvous problem assume that agents move synchronously in a network and that agents cannot leave any information on nodes, that is, whiteboards do not exist~\cite{Pelc2019book}.
These works have studied the feasibility of the rendezvous and, if feasible, the time required to accomplish the task.
If agents are anonymous (i.e., do not have IDs), the deterministic rendezvous cannot be achieved in some symmetric graphs because the symmetry cannot be broken.
In the literature~\cite{Dessmark2006,Kowalski2008,Ta-shma2007,Miller2016}, rendezvous algorithms have been proposed in any graph by assuming a unique ID for each agent.
Dessmark et al.~\cite{Dessmark2006} have proposed an algorithm to achieve the rendezvous in polynomial time of $n$, $\lambda$, and $\tau$, where $n$ is the number of nodes, $\lambda$ is the smallest ID among agents, and $\tau$ is the difference between the startup times of agents.
Kowalski et al.~\cite{Kowalski2008} and Ta{-}shma et al.~\cite{Ta-shma2007} have improved the time complexity and have proposed algorithms to achieve the rendezvous in time independent of $\tau$.
In addition, Millar et al.~\cite{Miller2016} have analyzed the trade-off between the time required for rendezvous and the number of moves.
On the other hand, some papers \cite{Fraigniaud2008,Fraigniaud2013,Czyzowicz2012htmwyf} have investigated the memory space, the time, and the number of moves required to achieve the deterministic rendezvous without assuming a unique ID of each agent.
Since the rendezvous cannot be accomplished for some initial arrangements of agents and graphs, they have proposed algorithms for limited graphs and initial arrangements.
Fraigniaud et al.~\cite{Fraigniaud2008,Fraigniaud2013} have proposed algorithms for trees, and Czyzowicz et al.~\cite{Czyzowicz2012htmwyf} have proposed an algorithm for arbitrary graphs when initial arrangements of agents are not symmetric.

While many papers deal with the rendezvous problem in synchronous environments, some papers assume asynchronous environments where agents move at different constant speeds or move asynchronously.
In the latter case, speeds of agents in each time are always determined by the adversary.
For more details, please refer to the literature~\cite{Marco2006,Guilbault2013,Dieudonne2015,Kranakis2017} for a finite graph and the literature~\cite{Czyzowicz2012htmae,Bampas2010,Collins2010} for an infinite graph.

Recently some papers~\cite{Dieudonne2014,Bouchard2016,Tsuchida2018btgoma,Tsuchida2018gomaiabe,Bouchard2018} have studied the gathering problem in the presence of Byzantine agents.
Table~\ref{tab_related_works} shows this research and the related researches that are closest to this research.
These studies assume agents with unique IDs and consider two types of Byzantine agents depending on whether they can falsify their own IDs.
\emph{Weakly Byzantine agents} perform arbitrary behaviors except falsifying their own IDs, and \emph{strongly Byzantine agents} perform arbitrary behaviors, including falsifying their own IDs.

Dieudonn{\'{e}} et al.~\cite{Dieudonne2014} have studied the gathering problem in synchronous environments
where $k$ agents exist in a $n$-node arbitrary network and $f$ of them are Byzantine.
For weakly Byzantine agents, if $n$ is given to agents, the gathering algorithm with the time complexity of $O(n^4\cdot |\Lambda_{good}|\cdot X(n))$ has been proposed, where $|\Lambda_{good}|$ is the length of the largest ID among non-Byzantine agents and $X(n)$ is the number of rounds required to explore any network composed of $n$ nodes, while, if $f$ is given to agents, the gathering algorithm with the time complexity that is polynomial of $n$ and $|\Lambda_{good}|$ has been proposed.
The numbers of non-Byzantine agents required for the gathering algorithms are at least one and $f+2$, respectively, and the numbers are proven to be tight.
On the other hand, for strongly Byzantine agents, in the cases where $n$ and $f$ are given to agents and $f$ is given to agents, the gathering algorithms whose time complexities are exponential of $n$ and $|\Lambda_{good}|$ have been proposed.
The numbers of non-Byzantine agents required for the gathering algorithms are at least $2f+1$ and $4f+2$, respectively, while the numbers of non-Byzantine agents required to solve the gathering problems under these conditions are $f+1$ and $f+2$, respectively.
Bouchard et al. \cite{Bouchard2016} have proposed the algorithms that show tight results for the number of non-Byzantine agents required to solve the gathering problem for both cases in the presence of strongly Byzantine agents.
That is, the numbers of non-Byzantine agents required for the algorithms are at least $f+1$ and $f+2$, respectively.
However, the time complexities of the algorithms are still exponential of $n$ and $|\Lambda_{good}|$.
Bouchard et al.~\cite{Bouchard2018} have proposed the gathering algorithm with the time complexity that is polynomial time for the first time in presence of strongly Byzantine agents in synchronous environments.
The gathering algorithm operates under the assumption that $\lceil \log \log n \rceil$ is given to agents and at least $5f^2+6f+2$ non-Byzantine agents exist in the network.

Tsuchida et al.~\cite{Tsuchida2018btgoma} have studied the gathering algorithm in synchronous environments with weakly Byzantine agents under the assumption that each node is equipped with an authenticated whiteboard, where each agent can leave information on its dedicated area but every agent can read all information.
If the upper bound $F$ of $f$ is given to agents, the gathering algorithm with the time complexity of $O(Fm)$ has been proposed, where $m$ is the number of edges.
Tsuchida et al.~\cite{Tsuchida2018gomaiabe} have proposed the gathering algorithms in asynchronous environments in the presence of weakly Byzantine agents under the same assumption of authenticated whiteboards.

\subsection{Our contributions}
\label{sec:our-results}
We seek an algorithm that achieves the gathering with small time complexity in synchronous environments with weakly Byzantine agents.
When agents cannot leave any information on nodes, the existing fastest algorithm is the one proposed by Dieudonn{\'{e}} et al.~\cite{Dieudonne2014}.
The algorithm tolerates any number of weakly Byzantine agents, achieves the gathering \emph{with simultaneous termination},
and its time complexity is $O(n^4\cdot |\Lambda_{good}|\cdot X(n))$, where $n$ is the number of nodes, $|\Lambda_{good}|$ is the length of the largest ID among non-Byzantine agents, and $X(n)$ is the number of rounds required to explore any network composed of $n$ nodes.
When agents can use authenticated whiteboards on nodes,
Tsuchida et al.~\cite{Tsuchida2018btgoma} have proposed the algorithm that is faster than that of Dieudonn{\'{e}} et al.~\cite{Dieudonne2014}.
However, the assumptions of authenticated whiteboards are strong and greatly restrict the behavior of Byzantine agents.

In this paper, we try to reduce the time complexity by taking advantage of a \emph{strong team}, that is, a team with a few Byzantine agents.
Since not so many agents are subject to faults in practice, the assumption of a strong team is reasonable.
We propose two gathering algorithms that tolerate $f$ weakly Byzantine agents in the case where a strong team composed of at least $4f^2+9f+4$ agents exist (see Table \ref{tab_related_works}).
Both the algorithms take the upper bound $N$ of $n$ as input.
The first algorithm achieves the gathering \emph{with non-simultaneous termination} and its time complexity is $O((f+|\Lambda_{good}|)\cdot X(N))$, where $|\Lambda_{good}|$ is the length of the maximum ID of non-Byzantine agents.
The second algorithm achieves the gathering \emph{with simultaneous termination} and its time complexity is $O((f+|\Lambda_{all}|)\cdot X(N))$, where $|\Lambda_{all}|$ is the length of the maximum ID of all agents.
If $n$ is given to agents, the second algorithm significantly reduces the time complexity compared to that of Dieudonn{\'{e}} et al.~in case of $|\Lambda_{all}| = O(|\Lambda_{good}|)$.

\section{Preliminaries}
\subsection{Distributed systems}
A distributed system is modeled by a connected undirected graph $G=(V,E)$, where $V$ is a set of $n$ nodes, and $E$ is a set of edges.
If an edge $\{u,v\} \in E$ exists between the nodes $u,v\in V$, $u$ and $v$ are said to be adjacent.
A set of adjacent nodes of node $v$ is denoted by $N_v=\{u\mid\{v,u\}\in E\}$.
The degree of node $v$ is defined as $d(v)=|N_v|$.
Each edge connected to node $v$ is locally and uniquely labeled by function $P_v:\{\{v,u\} \mid u\in N_v\}\rightarrow\{1,2,...,d(v)\}$ that satisfies $P_v(\{v,u\})\neq P_v(\{v,w\})$ for edges $\{v,u\}$ and $\{v,w\}$ $(u\neq w)$.
$P_v(v,u)$ is called the port number of an edge $\{v,u\}$ on node $v$.
Any node has neither ID nor memory.
Time is discretized, and each discretized time is called a round.

\subsection{Mobile agents}
There are $k$ agents $a_1, a_2, ..., a_k$ in the system.
All agents cannot mark visited nodes or traversed edges in any way.
Each agent $a_i$ has a unique ID denoted by $a_i.ID\in \mathbb{N}$, but does not know a priori the IDs of other agents.
Also, agents know the upper bound $N$ of the number of nodes, but they do not know $k$, the topology of the graph, or $n$.
The amount of agent memory is unlimited, and the contents of memory are not changed during a move through an edge.

The adversary wakes up at least one agent at the first round.
We call an agent that did not start at the first round dormant.
A dormant agent is woken up when the adversary wakes up the agent at some round or an agent visits the starting node of the dormant agent.
Note that the adversary can awake dormant agents at different rounds.

An agent is modeled as a state machine $(S,\delta)$.
Here, $S$ is a set of agent states, and a state is represented by a tuple of the values of all the variables that an agent has.
The state transition function $\delta$ outputs the next agent state, whether the agent stays or leaves, and the outgoing port number if the agent leaves.
The outputs are determined from the current agent state, the states of other agents on the same node, the degree of the current node, and the entry port.
An agent has a special state representing the termination of an algorithm.
After reaching the state, the agent never executes the algorithm.
If several agents are on node $v$, the agents can read all the information that they have (even if some of them have terminated).
However, if two agents traverse the same edge simultaneously in different directions, the agents do not notice this fact.
When an agent enters a node $v$ via an edge $\{u,v\}$, it learns the degree $d(v)$ of $v$ and the port number $P_v(v,u)$.
Agents execute the algorithm synchronously.
That is, at the beginning of a round, each agent reads states of all agents on the current node, executes the state transition.
If an agent decides to move, it arrives at the destination node before the beginning of the next round.
Note that, in each round, all agents on a single node obtain the same information of states of the agents.

\subsection{Byzantine agents}
There are $f$ \emph{weakly Byzantine} agents among $k$ agents.
Weakly Byzantine agents act arbitrarily without following an algorithm, but except changing their IDs.
All agents except weakly Byzantine agents are called \emph{good}.
Good agents know neither the actual value nor the upper bound of $f$.
The adversary wakes up at least one good agent at the first round.

\subsection{The gathering problems}
We consider the following two problems.
The gathering problem \emph{with non-simultaneous termination} requires the following conditions:
(1) every good agent terminates an algorithm, and
(2) when all the good agents terminate an algorithm, they are on the same node.
The gathering problem \emph{with simultaneous termination} requires all the good agents to terminate an algorithm at the same round on the same node.

We measure the time complexity of a gathering algorithm by the number of rounds from beginning (i.e., the first good agent wakes up) to the round in which all the good agents terminate.

\subsection{Procedures}
In the proposed algorithms, we use the graph exploration procedure and the extended label proposed in the literature.

The exploration procedure, called \textit{EXPLO}$(N)$, allows an agent to traverse all nodes of any graph composed of at most $N$ nodes, starting from any node of the graph.
An implementation of this procedure is based on universal exploration sequences (UXS) and is a corollary of the result by Reingold \cite{Reingold2008}.
The number of moves of \textit{EXPLO}$(N)$ is denoted by $X_N$.

Let $b_1b_2\cdots b_\ell$ be the binary representation of $a_i.ID$, where $\ell = |a_i.ID|$.
The extended label of $a_i$ is defined as $a_i.ID^*$ $=10b_1b_1b_2b_2$ $\cdots$ $b_\ell b_\ell$ $10b_1b_1b_2b_2$ $\cdots$ $b_\ell b_\ell \cdots$.
We have the following lemma about the extended label $a_i.ID^*$, which is used to prove the correctness of the proposed algorithms.
\begin{lemma}
\label{lemma_ExtendedLabels}
\cite{Dessmark2006}
For two different agents $a_i$ and $a_j$, assume that $a_i.ID^*=x_1x_2\cdots$ and $a_j.ID^*=y_1y_2\cdots$ hold.
Then, for some $k\leq 2\lfloor \log (\min(a_i.ID,a_j.ID)) \rfloor +6$, $x_k\neq y_k$ holds.
\end{lemma}

\section{A gathering algorithm with non-simultaneous termination}
\label{sec:notdetect}
In this section, we propose an algorithm for the gathering problem \emph{with non-simultaneous termination} by assuming a strong team composed of $4f^2+9f+4$ agents. That is, at least $(4f+4)(f+1)$ good agents exist in the network.
Recall that agents know $N$, but do not know $n$, $k$, or $f$.

\subsection{Overview}
The proposed algorithm aims to gather all good agents on a single node.
The algorithm achieves this goal by three stages: \CIST, \MGST, and \GST stages.
In the \CIST stage, agents collect IDs of all good agents.
In the \MGST stage, agents make a \emph{reliable group}, which is composed of at least $4f+4$ agents.
In the \GST stage, all good agents gather on a single node and achieve the gathering.
Each stage consists of multiple phases, and each phase consists of $P_N\ge X_N$ rounds.
We will discuss the actual value of $P_N$ later, and here just note that the duration of each phase is sufficient for an agent to explore the network by $\textit{EXPLO}(N)$.
For simplicity, we first explain the overview under the assumption that agents know $f$ and agents awake at the same round.
Under this assumption, all good agents start each phase at the same round.

In the \CIST stage, agents collect IDs of all good agents.
To do this, in the $x$-th phase of the \CIST stage, each agent $a_i$ reads the $x$-th bit of $a_i.ID^*$ and decides the behavior. 
If the bit is 1, $a_i$ executes \textit{EXPLO}$(N)$ during the phase.
If the bit is 0, $a_i$ waits during the phase.
Agent $a_i$ has variable $a_i.\IL$ to store a set of IDs, and if $a_i$ finds another agent on the same node while exploring or waiting, it records the agent's ID in $a_i.\IL$.
Agent $a_i$ executes this procedure until the $(2\lfloor\log (a_i.ID)\rfloor +6)$-th phase,
and then finishes the \CIST stage.
From Lemma \ref{lemma_ExtendedLabels}, $a_i$ can meet all other good agents and hence obtain IDs of all good agents.

In the \MGST stage, agents make a reliable group composed of at least $4f+4$ agents.
To do this, agents with small IDs keep waiting, and the other agents search for the agents with small IDs.
More concretely, if the $f+1$ smallest IDs in $a_i.\IL$ contains $a_i.ID$, $a_i$ keeps waiting during this stage.
Otherwise, $a_i$ assigns the smallest ID in $a_i.\IL$ to variable $a_i.\TAR$, and searches for the agent with ID $a_i.\TAR$, say $a_{target}$, by executing \textit{EXPLO}$(N)$ in a phase.
If $a_i$ finds $a_{target}$ on some node, it ends the search and waits on the node.
If $a_i$ does not find $a_{target}$ even after completing \textit{EXPLO}$(N)$, it regards $a_{target}$ as a Byzantine agent. 
In this case, $a_i$ assigns the second smallest ID in $a_i.\IL$ to $a_i.\TAR$, and searches for the agent with ID $a_i.\TAR$ in the next phase.
Agent $a_i$ continues this behavior until it finds a target agent.
Since there are at most $f$ Byzantine agents, the good agent with the smallest ID, say $a_{min}$, keeps waiting during the \MGST stage.
This means that agents always find $a_{min}$ if they search for $a_{min}$, and consequently, the number of agents searched for by good agents is at most $f+1$ (including $a_{min}$ and $f$ Byzantine agents).
Since at least $(4f+4)(f+1)$ good agents exist, even if the good agents are distributed to $f+1$ nodes evenly, at least $4f+4$ agents gather in one node according to the pigeonhole principle.
In other words, agents can make a reliable group.
The ID of the target agent in a reliable group is used as the group ID.
For \GST stage, a reliable group is divided into two groups, an exploring group and a waiting group, so that each of which contains at least $2f+2$ agents.

In the \GST stage, agents achieve the gathering after at least one reliable group is created. 
To do this, agents collect group IDs of all reliable groups in the first phase of the \GST stage. 
More concretely, while agents in a waiting group keep waiting, other agents (in an exploring group or not in a reliable group) explore the network by \textit{EXPLO}$(N)$.
When $a_i$ finds a reliable group, it records the group ID.
Note that, since each of an exploring group and a waiting group contains at least $2f+2$ agents, it contains at least $f+2$ good agents.
Therefore, when an agent meets an exploring or waiting group, the agent can understand that this group contains at least two good agents, and hence it is reliable.
In the second phase of the \GST stage, agents move to the node where the waiting group of the smallest group ID stays.
That is, while agents in the waiting group of the smallest group ID keep waiting, other agents search for the group by \textit{EXPLO}$(N)$.

However, there are three problems to implement the above behavior.
The first problem is that agents not in a reliable group cannot instantly know the fact that a reliable group has been created, and so they do not know when to transition to the \GST stage.
To solve this problem, we make agents execute the \MGST stage and the \GST stage alternately. 
Here, we design the two stages so that (1) agents achieve the gathering in the \GST stage if a reliable group is created in the \MGST stage, and (2) otherwise behaviors in the \GST stage do not affect the \MGST stage.
The second problem is that agents do not know $f$.
To solve this problem, at the end of the \CIST stage, agents estimate the number of Byzantine agents, say $\ESTF$, from the fact that at least $(4f+4)(f+1)$ good agents exist and their ID lists include IDs of all good agents.
However, values of $\ESTF$ differ by at most one among good agents, because some good agents may meet some Byzantine agents but others may not in the \CIST stage.
Therefore, we design the behaviors of the \MGST stage and the \GST stage so that agents can gather even if the estimated values have the difference.
The third problem is that some agents may be dormant.
To solve this problem, we make agents first explore the network by $\textit{EXPLO}(N)$ to wake up dormant agents.
As a result, we guarantee that all good agents start the algorithm within $X_N$ rounds, but there still exists a problem.
Good agents execute different phases at the same round because these agents woke up at different rounds.
So, we adjust the number of rounds of each phase to guarantee that all the good agents execute the same phase at the same time for sufficient rounds.

\subsection{Details}
\begin{figure}[t]
\begin{breakablealgorithm}
  \caption{Procedure Algorithm($N$) for an agent $a_i$ whose $a_i.ID=b_1b_2\cdots b_\ell$ where $\ell =|a_i.ID|$}
  \label{algo_entire}
  \begin{algorithmic}[1]
    \State $a_i.\STA \gets\SCI$
    \State $a_i.\IL \gets \{ a_i.ID\}$, $a_i.\BL \gets \emptyset$, $a_i.\GL \gets \emptyset$
    \State $a_i.\GID \gets NULL$
    \State $a_i.EndCI \gets False$
    \State $a_i.x \gets 1$
    
    \State Explore the network by \textit{EXPLO}$(N)$
    
    \While{$True$}
      \If{$a_i.EndCI=False$}
        \State Execute $a_i.x$-th phase of the \CIST stage
      \Else
        \State Execute the \MGST stage
      \EndIf
      \State $a_i.x\gets a_i.x+1$
      \State Execute the \GST stage
    \EndWhile
  \end{algorithmic} 
\end{breakablealgorithm} 
\end{figure}

\begin{figure}[t]
  \begin{center}
    \includegraphics[width=120mm]{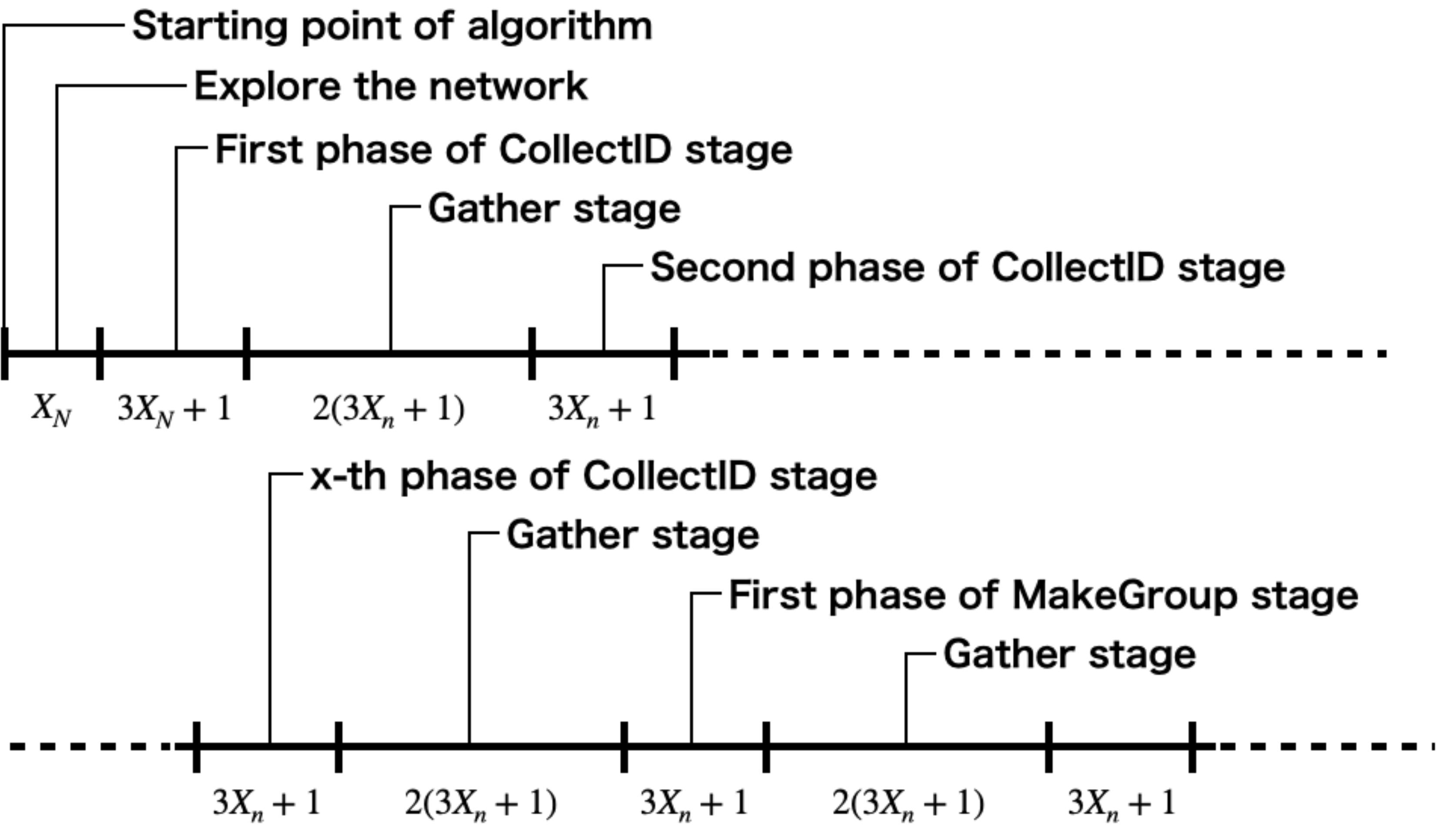}
    \caption{The stage flow.}
    \label{fig_stageflow}
  \end{center}
\end{figure}

\begin{table}[t]
  \centering
  \caption{Variables of agents.}
  \label{tab_valuelist}
  \begin{tabular}{cl}
    \hline
     Variable&\multicolumn{1}{c}{Explanation}\\\hline\hline
    $\STA$ & \begin{minipage}{120mm}
        \vspace{0.5mm}
        The current state of an agent.
        This variable takes one of the following values.
        \begin{itemize}
            \item $\SCI$
                (has not yet finished the \CIST stage)
            \item $\SMGSA$
                (works as a search agent in the \MGST stage)
            \item $\SMGTA$ 
                (works as a target agent in the \MGST stage)
            \item $\SGEG$ 
                (belongs to an exploring group in the \GST stage) 
            \item $\SGWG$ 
                (belongs to a waiting group in the \GST stage)
        \end{itemize}
        \end{minipage}\\ \hline
    $EndCI$ & \begin{tabular}{l}The variable that indicates whether an agent has finished the \CIST stage.\end{tabular} \\\hline
    $count$ & \begin{tabular}{l}The number of rounds from the beginning.\end{tabular} \\\hline
    $x$ & \begin{tabular}{l}The number of phases in the \CIST or  \MGST stage \end{tabular}\\\hline
    $\ESTF$ & \begin{tabular}{l}The estimated number of Byzantine agents.\end{tabular} \\\hline
    $\IL$ & \begin{tabular}{l}A set of agent IDs collected in the \CIST stage.\end{tabular} \\\hline
    $\BL$ & \begin{tabular}{l}A set of agent  IDs that the search agent regards as Byzantine agents.\end{tabular} \\\hline
    $\TAR$ & \begin{tabular}{l}Search agents: \begin{tabular}{l}The ID the agent searches for.\end{tabular}\\Target agents: \begin{tabular}{l}Its own ID.\end{tabular}\end{tabular} \\\hline
    $\GEF$ & \begin{tabular}{l}The consensus of $\ESTF$ among agents on the same node.\end{tabular} \\\hline
    $\GID$ & \begin{tabular}{l}The group ID of the reliable group that the agent belongs to.\end{tabular} \\\hline
    $\GL$ & \begin{tabular}{l}A set of group IDs collected in the \GST stage.\end{tabular} \\\hline
  \end{tabular}
\end{table}

Algorithm \ref{algo_entire} is the pseudocode of the proposed algorithm.
The proposed algorithm realizes the gathering using three stages:
The \CIST stage makes agents collect IDs of all good agents, the \MGST stage creates a reliable group composed of at least $4f+4$ agents, and the \GST stage gathers all good agents.

The overall flow of the algorithm is shown in Fig.\,\ref{fig_stageflow}.
After starting the algorithm, agent $a_i$ first explores the network with \textit{EXPLO}$(N)$ to wake up all dormant agents (line 6 of Algorithm \ref{algo_entire}).
By this behavior, after the first good agent wakes up, all good agents wake up within $X_N$ rounds.
After that, $a_i$ executes phases of the \CIST, \MGST, and \GST stages.
Here we define one phase as $P_N=3X_N+1$ rounds.
Since all good agents wake up within $X_N$ rounds, the $(X_N+1)$-th to $2X_N$-th rounds of the $x$-th phase of good agent $a_i$ overlap with the first $3X_N$ rounds of the $x$-th phases of all other good agents.
Hence, we have the following observation.

\begin{observation}
\label{observation_meet}
Let $a_i$ and $a_j$ be good agents. Assume that $a_i$ explores the network with $\textit{EXPLO}(N)$ from the $(X_N+1)$-th round to the $2X_N$-th round of its $x$-th phase, and $a_j$ waits during the first $3X_N$ rounds of its $x$-th phase. In this case, $a_i$ meets $a_j$ during the exploration.
\end{observation}

After the initial exploration, $a_i$ alternately executes one phase of the \CIST stage and two phases of the \GST stage (lines 9 and 14).
After $a_i$ finishes the \CIST stage, it alternately executes one phase of the \MGST stage (instead of the \CIST stage) and two phases of the \GST stage (lines 11 and 14).
The \GST stage interrupts the \CIST and \MGST stages, but, as described later, the behaviors of the \GST stage do not affect the behaviors of the \CIST and \MGST stages if no reliable group exists.
Therefore, we do not consider the behaviors of the \GST stage until a reliable group is created in the \MGST stage.

Table \ref{tab_valuelist} summarizes the variables used in the algorithm.
Agent $a_i$ stores the current state of $a_i$ in variable $a_i.\STA$.
Initially, $a_i.\STA = \SCI$ holds.
In addition, $a_i$ stores $False$ in variable $a_i.EndCI$ because it has not finished the \CIST stage.
Also, $a_i$ stores the number of rounds from the beginning in variable $a_i.count$. 
By variable $a_i.count$, $a_i$ determines which round of a phase it executes.
Agent $a_i$ increments $a_i.count$ for every round, but this behavior is omitted from the following description.

\subsubsection{The \CIST stage.}

\begin{figure}[t]
\begin{breakablealgorithm}
  \caption{The $a_i.x$-th phase of \CIST stage for an agent $a_i$}
  \label{CIST_code}
  \begin{algorithmic}[1]
    \If{the $a_i.x$-th bit of $a_i.ID^*$ is 0}
      \State Wait for $3X_N$ rounds on the current node
      \State $a_i.\IL \gets a_i.\IL \cup \{\textrm{IDs of agents $a_i$ met while waiting}\}$
    \Else
      \State Wait for $X_N$ rounds on the current node
      \State Explore the network by \textit{EXPLO}$(N)$
      \State Wait for $X_N$ rounds on the current node
      \State $a_i.\IL \gets a_i.\IL \cup \{\textrm{IDs of agents $a_i$ met while exploring}\}$
    \EndIf
    \State // The $(3X_N+1)$-th round
    \If{$a_i.x=2\lfloor \log a_i.ID \rfloor +6$}
      \State $a_i.\ESTF \gets \max\{y \mid (4y+4)(y+1)\leq |a_i.\IL|\}$
      \State $a_i.x\gets 1$
      \State $a_i.EndCI\gets True$
    \EndIf
    \State Wait for one round
  \end{algorithmic} 
\end{breakablealgorithm} 
\end{figure}

Algorithm \ref{CIST_code} is the pseudocode of the \CIST stage.
In the \CIST stage, agents collect IDs of all good agents.
The \CIST stage of $a_i$ consists of $2\lfloor\log (a_i.ID)\rfloor +6$ phases.
Note that the lengths of \CIST stages differ among agents.
Agent $a_i$ uses variable $a_i.\IL$ to store a set of IDs, and initially, it records $a_i.ID$ in $a_i.\IL$ (line 2 of Algorithm \ref{algo_entire}).
Agent $a_i$ determines the behavior of the $x$-th phase depending on the $x$-th bit of $a_i.ID^*$.
If the $x$-th bit is 0, $a_i$ waits for $3X_N$ rounds in the $x$-th phase (lines 1 to 2 of Algorithm \ref{CIST_code}).
If the $x$-th bit is 1, $a_i$ waits for $X_N$ rounds, explores the network by \textit{EXPLO}$(N)$, and then waits for $X_N$ round in the $x$-th phase (lines 4 to 7).
During these behaviors, if $a_i$ finds another agent $a_j$ on the same node, it records $a_j.ID$ in $a_i.\IL$ (lines 3 and 8).
Note that, from Lemma \ref{lemma_ExtendedLabels} and Observation \ref{observation_meet}, $a_i$ meets all good agents and records IDs of all good agents during the \CIST stage.

In the last round of the last phase of the \CIST stage, $a_i$ estimates the number of Byzantine agents $\ESTF$ as $a_i.\ESTF\gets \max\{y \mid (4y+4)(y+1)\leq |a_i.\IL|\}$ (line 12).
As we prove later, $a_i.\ESTF\geq f$ holds, and $|a_i.\ESTF-a_j.\ESTF|\leq 1$ holds for any good agent $a_j$.
Also, $a_i$ stores $True$ in $a_i.EndCI$ (line 14).

\subsubsection{The \MGST stage.}

\begin{figure}[!t]
\begin{breakablealgorithm}
  \caption{\MGST stage for an agent $a_i$}
  \label{MGST_code}
  \begin{algorithmic}[1]
    \If{$a_i.x=1$}
      \If{the smallest $a_i.\ESTF +1$ IDs in $a_i.\IL$ contain $a_i.ID$}
        \State $a_i.\STA \gets \SMGTA$
      \Else
        \State $a_i.\STA \gets \SMGSA$
      \EndIf
    \EndIf
    \If{$a_i.\STA =\SMGTA$}
      \State //$a_i$ is a target agent
      \State $a_i.\TAR \gets a_i.ID$
      \State Wait for one phase on the current node
      \State \textbf{and}
      \State While waiting, execute \textit{consensus}() every round
    \Else
      \State //$a_i$ is a search agent
      \State $a_i.\TAR \gets \min(a_i.\IL\setminus a_i.\BL)$
      \State Wait for $X_N$ rounds on the current node
      \State Search for an agent $a_{target}$ with ID $a_i.\TAR$ by \textit{EXPLO}$(N)$
      \State \textbf{and}
      \If{meet $a_{target}$ while searching}
        \State Stop \textit{EXPLO}$(N)$
        \State Wait until the end of the phase
        \State \textbf{and}
        \State While waiting, execute \textit{consensus}() every round
        \State \textbf{and}
        \If{$a_i$ finds $a_{target}$ Byzantine while waiting}
        \State // This is true if, during the $(X_N+1)$-th round to 
        \State // the $2X_N$-th round, $a_{target}$ moved to another
        \State // node or $a_{target}.\TAR\neq a_{target}.ID$ holds
          \State $a_i.\BL \gets a_i.\BL \cup \{ a_i.\TAR\}$  
        \EndIf
      \Else
      \State // Not meet $a_{target}$ and hence $a_{target}$ is Byzantine
      \State $a_i.\BL \gets a_i.\BL \cup \{ a_i.\TAR\}$
      \State Wait until the end of the phase
      \EndIf
    \EndIf
  \end{algorithmic} 
\end{breakablealgorithm} 
\end{figure}

\begin{figure}[!t]
\begin{breakablealgorithm}
  \caption{\textit{consensus()} for an agent $a_i$ (Compute the consensus of $\ESTF$ and determine whether a reliable group is created)}
  \label{Consensus_MGST_code}
  \begin{algorithmic}[1]
    \If{$a_i.GID = NULL$ and the number of agents in the \MGST stage on the current node is at least $4 \cdot a_i.\ESTF$}
      \State $a_i.\GEF\gets$ the most frequent value of $\ESTF$ of agents on the same node (if more than one most frequent value exists, choose the smallest one)
      \State Let $GC$ be a set of agents on the same node whose $\TAR$ is $a_i.\TAR$ and who execute the \MGST stage
      \If{$|GC|\geq 4\cdot a_i.\GEF +4$ and there exists $a_{target}$ with $a_{target}.\TAR=a_{target}.ID=a_i.\TAR$}
        \State $a_i.\GID \gets a_{target}.ID$
        \If{the $2\cdot a_i.\GEF +2$ smallest IDs in $GC$ contain $a_i.ID$}
          \State $a_i.\STA \gets \SGEG$
      \Else
          \State $a_i.\STA \gets \SGWG$
        \EndIf
      \EndIf
    \EndIf
  \end{algorithmic} 
\end{breakablealgorithm} 
\end{figure}

Algorithm \ref{MGST_code} is the pseudocode of the \MGST stage.
In the pseudo code, for simplicity we use \textbf{and} operation, which means that an agent executes the operations before and after the \textbf{and} operation at the same time.
In the \MGST stage, agents create a reliable group composed of at least $4f+4$ agents.
At the beginning of the \MGST stage, if the smallest $a_i.\ESTF+1$ IDs in $a_i.\IL$ contain $a_i.ID$, agent $a_i$ becomes a \emph{target agent} (line 3 of Algorithm \ref{MGST_code}).
Otherwise, $a_i$ becomes a \emph{search agent} (line 5).
Hereinafter, the good agent with the smallest ID is denoted by $a_{min}$.
As we prove later, $a_{min}$ always becomes a target agent.

If $a_i$ is a target agent, it executes $a_i.\TAR \gets a_i.ID$ (line 10) and waits for one phase on the current node (line 11). 
While waiting, $a_i$ executes procedure $\textit{consensus}()$ to determine whether a reliable group is created or not (line 13). 
We will explain the details of $\textit{consensus}()$ later.

Let us consider the case where $a_i$ is a search agent. 
The search agent $a_i$ stores in $a_i.\BL$ IDs of agents that $a_i$ regards as Byzantine agents (initially $a_i.\BL$ is empty).
In the first round of each phase, $a_i$ chooses the agent with the smallest ID, excluding Byzantine agents in $a_i.\BL$ (line 16).
After that, $a_i$ waits for $X_N$ rounds and then searches for the agent with ID $a_i.\TAR$, say $a_{target}$, by executing \textit{EXPLO}$(N)$ (lines 17 and 18).
If $a_i$ finds $a_{target}$ on the same node during the exploration, $a_i$ ends \textit{EXPLO}$(N)$ and waits on the node until the end of the phase (lines 21 to 22).
We can show that, if $a_{target}$ is good, $a_{target}$ keeps waiting as a target agent, and consequently, $a_i$ finds $a_{target}$ and waits with $a_{target}$.
Hence, if one of the following conditions holds, 
$a_i$ regards $a_{target}$ as a Byzantine agent: (1) $a_i$ did not find $a_{target}$ during the exploration (lines 33 to 34), or (2) after $a_i$ finds $a_{target}$, during the $(X_N+1)$-th round to the $2X_N$-th round, $a_{target}$ moved to another node or $a_{target}.\TAR\neq a_{target}.ID$ holds (lines 26 to 30).
In this case, $a_i$ adds $a_{target}.ID$ to $a_i.\BL$, and never searches for $a_{target}$ in the later phases of the \MGST stage (lines 30 and 34).
If $a_i$ did not find $a_{target}$, it waits until the end of the phase (line 35).

To determine whether agents can create a reliable group, search agents (resp., target agents) execute procedure $\textit{consensus}()$ in Algorithm \ref{Consensus_MGST_code} after they find their target agent (resp., from the beginning).
In procedure $\textit{consensus}()$, agent $a_i$ first calculates the consensus $a_i.\GEF$ of the estimated number of Byzantine agents as follows.
If the number of agents in the \MGST stage on the current node is at least $4\cdot a_i.\ESTF$, agent $a_i$ checks values of $\ESTF$ of all agents on the current node and assigns the most frequent value to $a_i.\GEF$ (line 2 of Algorithm \ref{Consensus_MGST_code}).
At this time, if multiple values are the most frequent, $a_i$ chooses the smallest one.

After that, $a_i$ determines whether a reliable group is created.
Agent $a_i$ observes states of all agents on the same node, and regards the set of agents whose $\TAR$ is $a_i.target$ and who execute the \MGST stage as the \emph{group candidate} (line 3).
If the group candidate contains at least $4\cdot a_i.\GEF+4$ agents and there exists $a_{target}$ with $a_{target}.\TAR=a_{target}.ID=a_i.\TAR$, $a_i$ regards the group candidate as a reliable group (line 4).
If $a_i$ understands that it is in a reliable group, $a_i$ stores $a_{target}.ID$ in variable $a_i.\GID$ as the group ID of the reliable group (line 5).
Note that, as we prove later, all other good agents in the reliable group also understand that they are in the reliable group, and assign $a_{target}.ID$ to their variable $\GID$ at the same round. 
Therefore, agents can identify members of a reliable group by observing variable $\GID$.
When a reliable group is created, the group is divided into two groups, a (reliable) \emph{exploring group} and a (reliable) \emph{waiting group}, for the \GST stage as follows.
If the $2\cdot a_i.\GEF+2$ smallest IDs among agents in $a_i$'s reliable group contains $a_i.ID$, $a_i$ belongs to an exploring group (line 7); otherwise, it belongs to a waiting group (line 9).
Note that each of an exploring group and a waiting group contains at least $2\cdot a_i.\GEF +2$ agents.

Once $a_i$ has determined that a reliable group is created, it does not calculate $a_i.\GEF$ and does not determine if a reliable group is created in subsequent rounds of this phase.
Note that some good agent $a_j$ with $a_j.\TAR=a_{target}.ID$ may visit the current node after $a_i$ determines a reliable group.
In this case, $a_j$ can become a member of the reliable group (i.e., $a_j.\GID \gets a_{target}.ID=a_i.\GID$).
This just increases the size of the reliable group and does not harm the algorithm.

\subsubsection{The \GST stage.}

\begin{figure}[!t]
\begin{breakablealgorithm}
  \caption{\GST stage for an agent $a_i$}
  \label{GST_code}
  \begin{algorithmic}[1]
    \If{$a_i.EndCI=False$}
      \State Wait for two phases on the current node
    \Else
      \State // The first phase
      \If{$a_i.\STA =\SGWG$}
        \State Wait for one phase on the current node
        \State \textbf{and}
        \State While waiting, whenever $a_i$ meets $a_j$ with $a_j.\GID \neq NULL$, execute $a_i.\GL \gets a_i.\GL \cup \{(a_j.\GID, a_j.ID)\}$
      \Else
        \State Wait for $X_N$ rounds on the current node
        \State Explore the network by \textit{EXPLO}$(N)$
        \State \textbf{and}
        \State While exploring, whenever $a_i$ meets $a_j$ with $a_j.\GID \neq NULL$, execute $a_i.\GL \gets a_i.\GL \cup \{(a_j.\GID, a_j.ID)\}$
        \State Wait for $X_N+1$ rounds on the current node
      \EndIf
    
      \State // The second phase
      \State //\textit{MemberID}$(gid)=\{id \mid (gid, id)\in a_i.\GL\}$
      \State //\textit{ReliableGID}$()=\{gid \mid |\textit{MemberID}(gid)|\geq a_i.\ESTF +1\}$
      \If{$\textit{ReliableGID}() =\emptyset$}
        \State Wait for one phase on the current node
      \ElsIf{$a_i.\STA =\SGWG$ and $a_i.\GID =\min(\textit{ReliableGID}())$}
        \State Wait for $3X_N$ rounds on the current node
        \State Terminate the algorithm
      \Else
        \State Wait for $X_N$ rounds on the current node
          \State By executing $\textit{EXPLO}(N)$, search for the node with a reliable waiting group whose group ID is $\min(\textit{ReliableGID}())$
          \State Wait on the node until the last round of the phase
          \State Terminate the algorithm at the last round of the phase%
      \EndIf
    \EndIf
  \end{algorithmic} 
\end{breakablealgorithm} 
\end{figure}

Algorithm \ref{GST_code} is the pseudocode of the \GST stage.
In the \GST stage, agents achieve the gathering if at least one reliable group exists in the network.
Note that two phases of the \GST stage interrupt phases of the \CIST and \MGST stages.
However, while executing the \GST stage, agents never update variables used in the \CIST and \MGST stages.
Also, recall that the behaviors of the \CIST and \MGST stages do not depend on the initial positions of agents in each phase.
Hence, the behaviors of the \GST stage do not affect the behaviors of the \CIST and \MGST stages.
If agents have not finished the \CIST stage, they wait for two phases (lines 1 to 2).
In the following, we describe the behaviors of agents that have finished the \CIST stage.

If agents have finished the \CIST stage, they try to achieve the gathering in two phases of the \GST stage.
In the first phase of the two phases, agents collect group IDs of all reliable groups (lines 4 to 15).
To do this, agents in waiting groups keep waiting for the phase, and other agents (agents in exploring groups and agents not in reliable groups) explore the network during the $(X_N+1)$-th round to the $2X_N$-th round.
During this behavior, when an agent finds a reliable waiting or exploring group, it records the group ID.
After that, in the second phase, they gather on the node where the reliable group with the smallest group ID exists (lines 16 to 29).

Here, we explain how agents find reliable exploring or waiting groups.
Since agents enter the \GST stage at different rounds, agents in a reliable group do not move together.
This implies that agent $a_i$ meets agents in a reliable group at different rounds.
For this reason, whenever agent $a_i$ meets $a_j$ with $a_j.\GID \neq NULL$ (i.e., $a_j$ says it is in a reliable group), $a_i$ adds a pair $(a_j.\GID,a_j.ID)$ in a set $a_i.\GL$.
Then, at the beginning of the second phase, $a_i$ checks $a_i.\GL$ and computes group IDs of reliable groups.
More concretely, $a_i$ determines that $gid$ is a group ID of a reliable group if there exist at least $a_i.\ESTF+1$ different IDs $id_1,id_2,\ldots$ such that $(gid,id_k)\in a_i.\GL$ for any $k$, that is, the number of agents that conveyed $gid$ as their group IDs is at least $a_i.\ESTF +1$.
In the rest of this paragraph, we explain why this threshold $a_i.\ESTF+1$ allows agent $a_i$ to recognize a reliable group correctly.
Assume that agent $a_i$ finds the exploring or waiting group that good agent $a_j$ belongs to.
Recall that the exploring or waiting group initially contains at least $2\cdot a_j.\GEF+2$ agents.
From this fact, even if $f\leq a_j.\GEF$ of them are Byzantine, at least $a_j.\GEF +2$ good agents convey their group ID to $a_i$.
Consequently, when $a_i$ finds the group, $a_i$ can determine that at least one good agent exists in this group because $|a_i.\ESTF -a_j.\GEF|\leq 1$ holds.
Therefore, if $a_i$ finds an exploring or waiting group (i.e., agents with the same $\GID$) composed of at least $a_i.\ESTF +1$ agents, $a_i$ can correctly recognize the group as a reliable group.

In the following, we explain the detailed behavior of agent $a_i$ in the two continuous phases of the \GST stage.

In the first phase, to collect all group IDs, agents in waiting groups keep waiting, and other agents (agents in exploring groups and agents not in reliable groups) explore the network.
To be more precise, if agent $a_i$ belongs to a reliable waiting group, $a_i$ collects pairs of a group ID and an agent ID in variable $a_i.\GL$ by waiting and observing visiting agents.
That is, $a_i$ waits for one phase, and if $a_i$ finds agent $a_j$ with $a_j.\GID\neq NULL$ while waiting, it adds $(a_j.\GID, a_j.ID)$ to $a_i.\GL$ (lines 6 to 8).
If agent $a_i$ belongs to a reliable exploring group or does not belong to a reliable group, $a_i$ collects pairs of a group ID and an agent ID in variable $a_i.\GL$ by exploring the network.
That is, $a_i$ waits for $X_N$ rounds, explores the network, and then waits for $X_N +1$ rounds.
If $a_i$ finds agent $a_j$ with $a_j.\GID\neq NULL$ during the exploration, it adds $(a_j.\GID, a_j.ID)$ to $a_i.\GL$ (lines 10 to 14).

In the second phase, all agents gather on the node where the reliable group with the smallest group ID exists.
Initially, $a_i$ calculates the set $\textit{ReliableGID}()$ of group IDs of all reliable groups as follows:
(1) $a_i$ makes, for each group ID $gid$ in $a_i.\GL$, a list of agent IDs that conveyed $gid$ as its group ID (i.e., \textit{MemberID}$(gid)=\{id \mid (gid, id)\in a_i.\GL\}$), and
(2) $a_i$ checks up group IDs such that at least $a_i.\ESTF+1$ agents conveyed the group ID (i.e., $\textit{ReliableGID}()=\{gid \mid |\textit{MemberID}(gid)| \geq a_i.\ESTF +1\}$).
Note that, if $a_i$ belongs to a reliable exploring (resp., waiting) group, $a_i.\GID\in\textit{ReliableGID}()$ holds because $a_i$ meets members of its own waiting (resp., exploring) group during the first phase.
If $a_i$ belongs to a reliable waiting group and satisfies $a_i.\GID=\min(\textit{ReliableGID}())$,
it waits for $3X_N$ rounds and terminates the algorithm (lines 21 to 23).
Otherwise, $a_i$ waits for $X_N$ rounds, and then, by executing \textit{EXPLO}$(N)$, searches for the node with the reliable waiting group whose group ID is $\min(\textit{ReliableGID}())$ (lines 25 to 26).
After that, $a_i$ waits until the last round of this phase and terminates the algorithm on the node (lines 27 to 28).

\subsection{Correctness and Complexity}
In this subsection, we prove correctness and complexity of the proposed algorithm.

\begin{lemma}
\label{lemma_AllGoodsKnowAllGoodsIDs}
Let $a_i$ be a good agent.
When $a_i$ finishes the \CIST stage, $a_i.\IL$ contains IDs of all good agents.
\end{lemma}
\begin{proof}
By Lemma \ref{lemma_ExtendedLabels} and Observation \ref{observation_meet}, $a_i$ meets all good agents before the end of the \CIST stage, and records their IDs in $a_i.\IL$.
Therefore, $a_i.\IL$ contains IDs of all good agents at the end of the \CIST stage.
\end{proof}

\begin{lemma}
\label{lemma_EstimateF_1}
After good agent $a_i$ finishes the \CIST stage, $a_i.\ESTF\geq f$ and $k\geq (4a_i.\ESTF +4)(a_i.\ESTF +1)$ hold.
\end{lemma}
\begin{proof}
By Lemma \ref{lemma_AllGoodsKnowAllGoodsIDs}, $a_i$ contains IDs of all good agents in $a_i.\IL$ at the end of \CIST stage, and so $|a_i.\IL|\geq (4f+4)(f+1)$ holds.
Therefore, we have $a_i.\ESTF=\max\{y\mid (4y+4)(y+1)\leq |a_i.\IL|\}\geq\max\{y\mid (4y+4)(y+1)\leq (4f+4)(f+1)\}=f$.
Also, by the algorithm, we clearly have $k\geq (4a_i.\ESTF +4)(a_i.\ESTF +1)$.
\end{proof}

\begin{lemma}
\label{lemma_EstimateF_2}
After good agents $a_i$ and $a_j$ finish the \CIST stage, $|a_i.\ESTF-a_j.\ESTF|\leq 1$ holds.
\end{lemma}
\begin{proof}
We prove this lemma by contradiction.
Without loss of generality, we assume $a_i.\ESTF=p$ and $a_j.\ESTF\geq p+2$.
We have $(4(p+1)+4)((p+1)+1)>|a_i.\IL|$ by $a_i.\ESTF<p+1$, and we have $(4(p+2)+4)((p+2)+1)\leq |a_j.\IL|$ by $a_j.\ESTF\geq p+2$.
Therefore, since $p\geq f$ holds by Lemma \ref{lemma_EstimateF_1}, $|a_j.L|-|a_i.L|>8p+20>f$ holds.
On the other hand, since $a_i.\IL$ and $a_j.\IL$ include IDs of all good agents by Lemma \ref{lemma_AllGoodsKnowAllGoodsIDs}, we have $|a_j.L|-|a_i.L|\leq f$, which contradicts the assumption.
\end{proof}

Let $\EFM$ be the largest value of $\ESTF$ among all good agents at the time when all good agents finish the \CIST stage.

\begin{lemma}
\label{lemma_NumTargetAgent}
The followings hold in the \MGST stage:
(1) $a_{min}$ is a target agent, and
(2) the number of good target agents is at most $\EFM +1$.
\end{lemma}
\begin{proof}
First, we prove proposition (1).
By Lemma \ref{lemma_EstimateF_1}, $a_{min}.\ESTF\geq f$ holds; thus, the $a_{min}.\ESTF+1$ $(\geq f+1)$ smallest IDs in $a_{min}.\IL$ contain $a_{min}.ID$. 
Therefore, $a_{min}$ is a target agent.

Next, we prove proposition (2) by contradiction.
Let us assume that proposition (2) does not hold.
That is, at least $\EFM +2$ good agents become target agents.
Let $a_{max}$ be the agent with the largest ID among the good target agents.
Since $a_{max}.\IL$ contains IDs of other $\EFM +1$ good agents that have smaller IDs than $a_{max}$,
$a_{max}$ does not become a target agent.
This is a contradiction.
Hence, the lemma holds.
\end{proof}

\begin{lemma}
\label{lemma_BlackList}
Let $a_i$ be a good agent.
Variable $a_i.\BL$ does not contain any ID of good agents.
\end{lemma}
\begin{proof}
We prove by induction.
Recall that $a_i$ adds $a_i.\TAR$ to $a_i.\BL$ in a phase of the \MGST stage only when one of the following conditions holds.
Let $a_{target}$ be the agent such that $a_i.\TAR=a_{target}.ID$ holds.
\begin{enumerate}
\item Agent $a_i$ did not find $a_{target}$ during the phase (line 34 of Alg.~\ref{MGST_code}).
\item After $a_i$ found $a_{target}$, during the $(X_N+1)$-th round to the $2X_N$-th round of the phase, $a_{target}$ moved to another node or $a_{target}.\TAR\neq a_{target}.ID$ holds (line 30 of Alg.~\ref{MGST_code}).
\end{enumerate}

For the base case, we consider the first phase of the \MGST stage of $a_i$.
By Lemma \ref{lemma_AllGoodsKnowAllGoodsIDs}, $a_i.\IL$ contains IDs of all good agents.
Since $a_i.\BL$ is empty at the beginning of the first phase, $a_i.\TAR$ $(= \min(a_i.\IL))$ is $a_{min}.ID$ or an ID of a Byzantine agent.
But, here, it is sufficient to consider only the former case.
Since $a_{min}$ has the smallest ID among good agents, the duration of the \CIST stage is the shortest among good agents.
Hence, $a_{min}$ starts the \MGST stage before $a_i$ starts the $(X_N+1)$-th round of the first phase of the \MGST stage.
Since $a_{min}$ is a target agent by Lemma \ref{lemma_NumTargetAgent}, $a_{min}$ continues to wait during the \MGST stage. %
This implies that the above conditions to update $a_i.\BL$ are not satisfied. 
Hence, $a_i$ does not update $a_i.\BL$, and the lemma holds in the first phase.

For the induction, assume that $a_i.\BL$ does not contain IDs of good agents at the end of the $t$-th phase of the \MGST stage of $a_i$.
We consider the $(t+1)$-th phase of the \MGST stage of $a_i$. 
Since $a_i.\BL$ does not contain IDs of the good agents at the beginning of the $(t+1)$-th phase, $a_i.\TAR = \min(a_i.\IL\setminus a_i.\BL)$ is $a_{min}.ID$ or an ID of a Byzantine agent.
By the same discussion as in the first phase, we can prove that IDs of good agents are not added to $a_i.\BL$ in the $(t+1)$-th phase.
Therefore, this lemma holds in the $(t+1)$-th phase.
Hence, the lemma holds.
\end{proof}

\begin{lemma}
\label{lemma_GoodsFSomeGoodsEF}
When good agent $a_i$ executes $a_i.\GEF \gets \ESTF'$ in $\textit{consensus}()$, there exists good agent $a_j$ with $a_j.\ESTF=\ESTF'$
\end{lemma}
\begin{proof}
Assume that $a_i$ executes $a_i.\GEF \gets \ESTF'$ on node $v$ in round $r$.
By the algorithm, in round $r$, there exist at least $4\cdot a_i.\ESTF$ agents executing the \MGST stage on node $v$.
Since $a_i.\ESTF\geq f$ holds by Lemma \ref{lemma_EstimateF_1}, there exist at least $4\cdot a_i.\ESTF -f\geq 4f-f=3f$ good agents executing the \MGST stage on $v$ in round $r$.
Also, since variable $\ESTF$ of good agents takes at most two possible values by Lemma \ref{lemma_EstimateF_2}, at least $\lceil 3f/2\rceil >f$ good agents on $v$ have the same value of $\ESTF$.
Therefore, in round $r$, $a_i$ stores the value of variable $\ESTF$ of some good agent in $a_i.\GEF$.
Hence, the lemma holds.
\end{proof}

\begin{lemma}
\label{lemma_ReliableGroupEstimateFAndGID}
If good agent $a_i$ determines that a reliable group is created on node $v$ in round $r$,
there exists a set $A'$ of agents that satisfies the following conditions:
\begin{itemize}
\item
Set $A'$ contains at least $4\cdot a_i.\GEF +4$ agents.
\item
Good agents in $A'$ determine that a reliable group is created on $v$ in round $r$.
\item
For any good agent $a_j$ in $A'$, $a_j.\GEF =a_i.\GEF$ and $a_j.\GID =a_i.\GID$ hold at the end of round $r$.
\end{itemize}
\end{lemma}
\begin{proof}
Assume that good agent $a_i$ determines that a reliable group is created on $v$ in round $r$.
Let $A'$ be a set of agents such that, iff $a_j\in A'$ holds, $a_j$ stays on $v$ in round $r$ and $a_j.\TAR=a_i.\TAR$ holds.
We prove that $A'$ satisfies the conditions of the lemma.
Since $a_i$ determines that a reliable group is created, $A'$ contains at least $4\cdot a_i.\GEF +4$ agents.
Also, $A'$ contains agent $a_{target}$ with $a_{target}.ID=a_i.\TAR$.
Fix an agent $a_j\in A'$.
By Lemmas \ref{lemma_EstimateF_2} and \ref{lemma_GoodsFSomeGoodsEF}, $a_j.\ESTF\leq a_i.\GEF+1$ holds, and hence, $4\cdot a_i.\GEF+4\geq 4\cdot a_j.\ESTF$ hold.
This implies that the number of agents on $v$ satisfies the condition that $a_j$ calculates $a_j.\GEF$ (line 1 of Algorithm \ref{Consensus_MGST_code}).
Since the situation of $v$ is the same for both $a_i$ and $a_j$, $a_j.\GEF=a_i.\GEF$ holds.
In addition, $a_j$ also observes agents in $A'$; then, $a_j$ determines that a reliable group is created on $v$ in round $r$.
Thus, $a_j$ executes $a_j.\GID \gets a_j.\TAR$ $(=a_i.\TAR)$, and $a_j.\GID =a_i.\GID$ holds.
Hence, the lemma holds.
\end{proof}

In the following two lemmas, we prove that a reliable group is created before all good agents finish the $(f+1)$-th phase of the \MGST stage. 
Let $a_{last}$ be the good agent that finishes the \CIST stage last,
and let $\textit{phase}_x$ be the $x$-th phase of the \MGST stage of $a_{last}$.
Since all agents wake up within $X_N$ rounds and each phase consists of $3X_N+1$ rounds, any good agent $a_i$ has exactly one phase $\textit{phase}^i_x$ that overlaps $\textit{phase}_x$ for at least $2X_N+1$ rounds.
For simplicity, when agent $a_i$ behaves in $\textit{phase}^i_x$, we say that $a_i$ behaves in the $x$-th phase (of the \MGST stage) of $a_{last}$.

\begin{lemma}
\label{lemma_GoodBelongToByzTarget}
Let $Byz_1,Byz_2,\ldots,Byz_{f'}$ ($Byz_l.ID < Byz_{l+1}.ID$ for $1\leq l\leq f'-1$) be Byzantine agents whose IDs are smaller than $a_{min}$.
Assume that, when $a_{last}$ finishes the $f'$-th phase of the \MGST stage, a reliable group does not exist.
Then, in the $(f'+1)$-th phase of the \MGST stage of $a_{last}$, at most $(4\EFM +2)f'$ good agents assign $bid\in\{Byz_1.ID,Byz_2.ID,\ldots,$ $Byz_{f'}.ID\}$ to their variable $\TAR$. 
\end{lemma}
\begin{proof}
Assume that a reliable group does not exist when $a_{last}$ finishes the $f'$-th phase of the \MGST stage.
Under this assumption, we prove by induction that, in the $(x+1)$-th phase of the \MGST stage ($1 \leq x \leq f'$) of $a_{last}$, at most $(4\EFM +2)x$ good agents assign $bid\in\{Byz_1.ID,Byz_2.ID,\ldots,Byz_{x}.ID\}$ to their variable $\TAR$.
Hereinafter, the $x$-th phase of the \MGST stage of $a_{last}$ is simply called the $x$-th phase.

For the base case, we consider the case of $x=1$.
Let $A_1$ be a set of good agents that assign $Byz_1.ID$ to their variable $\TAR$ in the second phase.
For contradiction, assume $|A_1|>4\EFM +2$.
Since good agents monotonically increase $\TAR$, agents in $A_1$ also assign $Byz_1.ID$ to $\TAR$ in the first phase.
Also, since the agents do not regard $Byz_1$ as a Byzantine agent in the first phase, they find $Byz_1$ in the first phase and, after that, $Byz_1$ does not move and $Byz_1.\TAR=Byz_1.ID$ holds until the $2X_N$-th round of the first phase.
In addition, they start the first phase within at most $X_N$ round and wait during the $(2X_N+1)$-th round to the $(3X_N+1)$-th round of the first phase.
This implies that all agents in $A_1$ exist on the same node as $Byz_1$ before the $2X_N$-th round of the first phase, and at that time the number of agents on the node is at least $4\EFM +4$.
This contradicts the assumption since a reliable group is created by the algorithm.
Therefore, $|A_1|\leq 4\EFM +2$ holds.

For induction step, assume that, in the $(x+1)$-th phase ($1\leq x < f'$), at most $(4\EFM +2)x$ good agents assign $bid\in\{Byz_1.ID,Byz_2.ID,$ $\ldots,Byz_x.ID\}$ to their $\TAR$.
Let $A_x$ be a set of good agents that assign $bid\in\{Byz_1.ID,Byz_2.ID,\ldots,Byz_{x+1}.ID\}$ to $\TAR$ in the $(x+2)$-th phase.
For contradiction, assume $|A_x|>(4\EFM +2)(x+1)$.
Let $B_x$ be a set of good agents that assign $Byz_{x+1}.ID$ to $\TAR$ in the $(x+1)$-th phase, and let $C_x$ be a set of good agents that assign $bid\in\{Byz_1.ID,Byz_2.ID,\ldots,Byz_x.ID\}$ to $\TAR$ in the $(x+1)$-th phase.
Since good agents monotonically increase $\TAR$, $A_x\subseteq B_x\cup C_x$ holds.
Since $|C_x|\leq (4\EFM +2)x$ holds by the assumption of induction, $|B_x \cap A_x| \geq |A_x|-|C_x|>4\EFM +2$ holds.
Since good agents in $B_x \cap A_x$ do not regard $Byz_{x+1}$ as a Byzantine agent in the $(x+1)$-th phase, they find $Byz_{x+1}$, and, after that, $Byz_{x+1}$ does not move and $Byz_{x+1}.\TAR =Byz_{x+1}.ID$ holds until the $2X_N$-th round of the $(x+1)$-th phase.
Similarly to the base case, this implies that all agents in $B_x \cap A_x$ exist on the same node as $Byz_{x+1}$, and at that time, the number of agents on the node is at least $4\EFM +4$.
This contradicts the assumption since a reliable group is created by the algorithm.
Therefore, $|A_x|\leq (4\EFM +2)(x+1)$ holds.

Hence, the lemma holds.
\end{proof}

\begin{lemma}
\label{lemma_NumMTPhase}
Before $a_{last}$ finishes the $(f+1)$-th phase of the \MGST stage, a reliable group is created.
\end{lemma}
\begin{proof}
Let $f'(\leq f)$  be the number of Byzantine agents whose IDs are smaller than $a_{min}.ID$.
By Lemma \ref{lemma_GoodBelongToByzTarget}, if a reliable group is not created before $a_{last}$ finishes the $f'$-th phase of the \MGST stage, at most $(4\EFM +2)f'$ good agents assign an ID of a Byzantine agent with a smaller ID than $a_{min}$ to $\TAR$ in the $(f'+1)$-th phase of $a_{last}$.
Also, by Lemma \ref{lemma_NumTargetAgent}, the number of good target agents is at most $\EFM +1$.
This implies that, in the $(f'+1)$-th phase of $a_{last}$, at least $(k-f)-(\EFM+1)-(4\EFM+2)f'$ good search agents assign $a_{min}.ID$ to $\TAR$ (because $a_{min}.ID$ is not in variable $\BL$ of agents by Lemma \ref{lemma_BlackList}).
Since they can successfully find $a_{min}$, by Lemma \ref{lemma_EstimateF_1}, at least $(k-f)-(\EFM +1)-(4\EFM +2)f'\geq (4\EFM +4)(\EFM +1)-\EFM-(\EFM+1)-(4\EFM+2)\EFM = 4\EFM +3$ search agents stay with target agent $a_{min}$ before the $2X_N$-th rounds of the $(f'+1)$-th phase of $a_{last}$.
This implies that they make a reliable group.
Hence, the lemma holds.
\end{proof}

The following two lemmas show that agents can achieve the gathering if at least one reliable group is created and they finish the \CIST stage.
Let $a_{ini}$ be the good agent that wakes up earliest.
Since all agents wake up within $X_N$ rounds, if $a_{ini}$ starts two consecutive phases of the \GST stage in round $r$, all good agents start two consecutive phases of the \GST stage before round $r+X_N$.

\begin{lemma}
\label{lemma_ReliableGIDFunc}
Consider the following situation: 
(1) $a_{ini}$ starts two consecutive phases of the \GST stage in round $r$, 
(2) $a_i$ (possibly $a_{ini}$) starts two consecutive phases of the \GST stage in round $r'$ such that $r\leq r' \leq r+X_N$ holds, and 
(3) $a_i$ has completed the \CIST stage before round $r'$.
Let $List_i$ be the output of $\textit{ReliableGID()}$ for $a_i$ in the two consecutive phases, and let $Rel$ be a set of reliable groups that exist in round $r+X_N$.
Then, $List_i$ is a set of all group IDs of $Rel$.
\end{lemma}
\begin{proof}
By the algorithm, since all good agents wake up within $X_N$ rounds, all good agents start two consecutive phases of the \GST stage during rounds $r$ to $r+X_N$ and hence, no new reliable group is created during rounds $r+X_N$ to $r+2X_N$.

If $a_i$ belongs to a reliable waiting group, it waits during rounds $r'(\leq r+X_N)$ to $r'+3X_N (\geq r+3X_N)$.
Since all good agents in reliable exploring groups explore the network during rounds $r+X_N$ to $r+3X_N$, all of them meet $a_i$.
Therefore, for each good agent $a$ in a reliable exploring group of $Rel$, $a_i.\GL$ contains $(a.GID,a.ID)$.

If $a_i$ does not belong to a reliable waiting group, it explores the network during rounds $r'+X_N (\geq r+X_N)$ to $r'+2X_N(\leq r+3X_N)$.
Since all good agents in reliable waiting groups wait during rounds $r+X_N$ to $r+3X_N$, all of them meet $a_i$.
Therefore, for each good agent $a$ in a reliable waiting group of $Rel$, $a_i.\GL$ contains $(a.GID,a.ID)$.

Let $a_k$ be a good agent that belongs to a group in $Rel$.
By Lemma \ref{lemma_ReliableGroupEstimateFAndGID}, the reliable group of $a_k$ contains at least $4\cdot a_k.\GEF +4-f$ good agents, and hence, each of the exploring group and the waiting group contains at least $2\cdot a_k.\GEF +2-f\geq a_k.\GEF+2$ good agents.
By Lemmas \ref{lemma_EstimateF_2} and \ref{lemma_GoodsFSomeGoodsEF},  since $a_k.\GEF+2\geq a_i.\ESTF +1$ holds from $|a_k.\GEF-a_i.\ESTF|\leq 1$, $a_i.\GL$ contains at least $a_i.\ESTF+1$ pairs for each group in $Rel$.
Hence, $List_i$ contains all group IDs of $Rel$.
In addition, since there exist $f$ Byzantine agents, $List_i$ does not contain a fake group ID that was conveyed by Byzantine agents.
Hence, $List_i$ is a set of all group IDs of $Rel$.

\end{proof}

\begin{lemma}
\label{lemma_ReliableGroupCompleteTask}
Let $r$ be the first round such that (a) $a_{ini}$ starts two consecutive phases of the \GST stage in round $r$ and (b) there exists a reliable group in round $r+X_N$.
Let $Rel$ be a set of reliable groups that exist in round $r+X_N$.
Let $G_{min}$ be the group with the smallest group ID $gid_{min}$ in $Rel$.
Let $v_{min}$ be the node where $G_{min}$ is created.
Assume that $a_i$ (possibly $a_{ini}$) starts two consecutive phases of the \GST stage in round $r'$ such that $r\le r'\le r+X_N$.
Then, the following propositions hold:
(1) If $a_i$ has finished the \CIST stage before round $r'$, it terminates the algorithm on $v_{min}$ during the two consecutive phases of the \GST stage after round $r'$.
(2) If $a_i$ has not finished the \CIST stage in round $r'$, it terminates the algorithm on $v_{min}$ in the first two consecutive phases of the \GST stage after it finishes the \CIST stage.
\end{lemma}

\begin{proof}
First, we prove proposition (1).
We focus on the first two consecutive phases of the \GST stage after round $r'$.
From Lemma \ref{lemma_ReliableGIDFunc}, $a_i$ obtains the set of all group IDs of $Rel$ as the output of $\textit{ReliableGID()}$ and hence, $\min(\textit{ReliableGID()})$ is $gid_{min}$.
Hence, if $a_i$ belongs to a reliable waiting group of $G_{min}$, it terminates on its current node $v_{min}$ at the $3X_N+1$ round of the second phase after round $r'$.
Otherwise, $a_i$ searches for the waiting group of $G_{min}$ in the second phase after round $r'$.
More concretely, $a_i$ explores the network during the $(X_N+1)$-th round to the $2X_N$-th round in the second phase.
Recall that agents in a reliable waiting group of $G_{min}$ wait $3X_N$ rounds before terminating on $v_{min}$ in their second phases, and the difference of starting times of the phases is at most $X_N$.
Hence, $a_i$ meets agents in a reliable waiting group of $G_{min}$ on $v_{min}$ during the exploration, and then, it terminates on $v_{min}$.

Next, we prove proposition (2).
Consider the case that $a_i$ is the first agent that finishes the \CIST stage after $r'$.
Assume that, in round $r''$, $a_i$ finishes the \CIST stage.
Since no agent executes the \MGST stage between $r'$ and $r$, the set of reliable groups is $Rel$.
Since all agents that belong to groups in $Rel$ have terminated from proposition (1), $a_i$ meets all of them in the first phase of the \GST stage after round $r''$.
Hence, in the second phase, $gid_{min}
=\min(\textit{ReliableGID()})$ holds, and consequently $a_i$ terminates the algorithm on $v_{min}$ during the second phase.
Consider the case that $a_i$ is not the first agent that finishes the \CIST stage after $r'$.
Even in this case, the set of reliable groups is still $Rel$.
Hence, we can prove this case similarly to the above case.

\end{proof}

Finally, we prove the complexity of the proposed algorithm.

\begin{theorem}
\label{theorem_NotDetectGathering}
Let $n$ be the number of nodes, $k$ be the number of agents, $f$ be the number of weakly Byzantine agents, and $\Lambda_{good}$ be the largest ID among good agents.
If the upper bound $N$ of $n$ is given to agents and $(4f+4)(f+1)\leq k$ holds, the proposed algorithm solves the gathering problem with non-simultaneous termination in at most $X_N+3(2\lfloor\log\Lambda_{good}\rfloor +f+7)(3X_N+1)$ rounds.
\end{theorem}
\begin{proof}
Let $a_{last}$ be the good agent that finishes the \CIST stage last.
Since $a_{last}$ wakes up within $X_N$ rounds (after the first agent wakes up) and executes at most $2\lfloor\log \Lambda_{good} \rfloor +6$ phases of the \CIST stage, $a_{last}$ finishes the \CIST stage in $X_N+(2\lfloor\log \Lambda_{good} \rfloor +6)\cdot 3(3X_N +1)=X_N+3(2\lfloor\log \Lambda_{good} \rfloor +6)(3X_N +1)$ rounds.
By Lemma \ref{lemma_NumMTPhase}, a reliable group is created before $a_{last}$ finishes the $(f+1)$-th phase of the \MGST stage.
By Lemma \ref{lemma_ReliableGroupCompleteTask}, if at least one reliable group is created and all good agents finish the \CIST stage, agents achieve the gathering during the next two phases of the \GST stage.
Therefore, agents achieve the gathering in at most $X_N+3(2\lfloor\log\Lambda_{good}\rfloor +6)(3X_N +1)+(f+1)\cdot 3(3X_N +1)=X_N+3(2\lfloor\log\Lambda_{good}\rfloor +f+7)(3X_N +1)$ rounds.
\end{proof}

\section{A gathering algorithm with simultaneous termination}
\label{sec:detect}
In this section, we propose an algorithm for the gathering problem \emph{with simultaneous termination} by modifying the algorithm in the previous section.
The underlying assumption is the same as that of the previous section.
In the following, we refer to the proposed algorithm in the previous section as the previous algorithm.
By the previous algorithm, all good agents gather on a single node but terminate at different rounds.
Therefore, the purpose of this section is to change the termination condition of the previous algorithm so that all good agents terminate at the same round.

By Lemma \ref{lemma_ReliableGroupCompleteTask}, after all good agents finish the \CIST stage and at least one reliable group is created, all good agents gather at a single node during the next two consecutive phases of the \GST stage.
Hence, after good agents move to the gathering node in the \GST stage, they can terminate at the same round if they wait until all good agents finish the \CIST stage (and the next \GST stage).
To do this, we can use the fact that, when good agent $a_i$ finishes the \CIST stage, $a_i.\IL$ contains IDs of all good agents.
That is, $\max(a_i.\IL)$ is the upper bound of IDs of good agents and hence, $a_i$ can compute the upper bound of rounds required for all good agents to finish the \CIST stage.
However, for two good agents $a_i$ and $a_j$, $\max(a_i.\IL)$ can be different from $\max(a_j.\IL)$ because it is possible that either $a_i$ or $a_j$ meets a Byzantine agent with an ID larger than the largest ID among good agents.
Also, if agents share their variable $\IL$ and take the maximum ID, Byzantine agents may share a very large ID such that no agent has the ID.
To overcome this problem, each agent $a_i$ selects the largest ID among IDs that $a_i.\GEF+1$ agents have in their variable $\IL$, and computes when to terminate.
Note that, in order that all good agents agree on the largest ID, they should have the same value of $\GEF$.
For this reason, each agent $a_i$ updates $a_i.\GEF$ similarly to the \MGST stage after it completes the previous algorithm.
Since all good agents in a reliable group exist on the gathering node, $a_i$ can correctly update $a_i.\GEF$.

Lastly, to terminate at the same round, good agents make a consensus on termination. 
To do this, each agent $a_i$ prepares a flag $a_i.flag_t$ (initially, $a_i.flag_t\gets False$). 
Agent $a_i$ executes $a_i.flag_t\gets True$ if it is ready to terminate, i.e., it understands that all good agents gather on the current node. 
After $a_i$ completes the previous algorithm, it also checks $flag_t$ of all agents on the current node every round.
If $flag_t$ of at least $a_i.\GEF+1$ agents are true, $a_i$ terminates the algorithm because at least one good agent understands that all good agents gather on the current node.
Since all good agents stay at the same node and make the decision based on the same information, they can terminate at the same round.

In this paragraph, we describe the detailed behavior of $a_i$ in the algorithm.
First, $a_i$ executes the previous algorithm until just before it terminates, but it does not terminate.
Let round $r_i$ be the round immediately after $a_i$ completes the previous algorithm.
After round $r_i$, $a_i$ waits on the gathering node of the previous algorithm, say $v$, and always checks whether it can terminate.
More concretely, $a_i$ executes the following operations every round after round $r_i$.
\begin{enumerate}
\item Agent $a_i$ updates $a_i.\GEF$ in the same way as in the \MGST stage of the previous algorithm, that is, $a_i$ assigns the most frequent value of $\ESTF$ to $a_i.\GEF$. If multiple values are the most frequent, $a_i$ chooses the smallest one.
\item Agent $a_i$ checks $flag_t$ of agents on $v$, and, if $flag_t$ of at least $a_i.\GEF+1$ agents are true, $a_i$ terminates the algorithm.
\item Agent $a_i$ checks variable $\IL$ of agents on $v$ and computes the maximum ID among agents.
That is, letting $\IL_g$ be a set of IDs that at least $a_i.\GEF+1$ agents on $v$ have in their variable $\IL$, $a_i$ executes $a_i.\IDM \gets \max(\IL_g)$.
\item Agent $a_i$ checks whether all good agents gather on $v$.
If all good agents have completed the \CIST stage before round $r_i$, all good agents gather on $v$ before round $r_i+X_N$ because all agents wake up within $X_N$ rounds.
Consider the case that some good agent has not yet completed the \CIST stage in round $r_i$.
Since a reliable group has already been created, if the agent with ID $a_i.\IDM$ has finished the \CIST stage and its next two phases of the \GST stage, $a_i$ understands that all good agents gather on $v$.
Note that the agent with ID $a_i.\IDM$ completes the \CIST stage and its next two phases of the \GST stage in at most $T=X_N+X_N+3(2\lfloor\log (a_i.\IDM)\rfloor+6)(3X_N+1)$ rounds after $a_i$ starts the algorithm.
For this reason, $a_i$ sets $a_i.flag_t\gets True$ if (a) $X_N$ rounds have elapsed after round $r_i$ and (b) $T$ rounds have elapsed after it starts the algorithm.
\end{enumerate}

\begin{theorem}
\label{theorem_DetectGathering}
Let $n$ be the number of nodes, $k$ be the number of agents, $f$ be the number of Byzantine agents, and $\Lambda_{all}$ be the largest ID among all agents.
If the upper bound $N$ of $n$ is given to agents and $(4f+4)(f+1)\leq k$ holds, the proposed algorithm solves the gathering problem with simultaneous termination in at most $3X_N+3(2\lfloor\log\Lambda_{all}\rfloor +f+7)(3X_N +1)+1$ rounds.
\end{theorem}

\begin{proof}
Let $a_{ini}$ be the agent that starts the algorithm earliest.
Let $r$ be the first round such that (a) $a_{ini}$ starts two consecutive phases of the \GST stage in round $r$ and (b) there exists a reliable group in round $r+X_N$, and let $Rel$ be a set of reliable groups that exist in round $r+X_N$.
Let $G_{min}$ be the group with the smallest group ID in $Rel$, and let $v_{min}$ be the node where $G_{min}$ is created.
From Lemma \ref{lemma_ReliableGroupCompleteTask}, each good agent exists on $v_{min}$ when it completes the previous algorithm.

Let $a_f$ be the agent that executes $flag_t\gets True$ earliest, and assume that $a_f$ executes $a_f.flag_t\gets True$ in round $r^*$.

First, we prove that all good agents complete the previous algorithm before round $r^*$.
Assume that $a_f$ completes the previous algorithm in round $r_f$.
If all good agents complete the \CIST stage before round $r_f$, all good agents gather on $v$ before round $r_f+X_N$.
Since $r^*\ge r_f+X_N$ holds, all good agents complete the previous algorithm before round $r^*$.
Consider the case that some good agent has not yet completed the \CIST stage in round $r_f$.
Since all agents wake up within $X_N$ rounds and agents do not move during the last $X_N$ rounds of the previous algorithm, good agents in a reliable group in $Rel$ exist on $v_{min}$ after round $r_f$.
Hence, at least $4 \cdot a_f.\GEF+4-f \geq 3f$ good agents exist on $v_{min}$ after round $r_f$.
Hence, similarly to Lemma \ref{lemma_GoodsFSomeGoodsEF}, $a_f$ assigns $\ESTF$ of some good agent to $a_f.\GEF$ after round $r_f$.
This implies that $a_f$ assigns an ID of some agent to $a_f.\IDM$.
Note that the assigned ID is at least $\Lambda_{good}$, where $\Lambda_{good}$ is the largest ID among all good agents.
Hence, since $a_f$ executes $flag_t\gets True$ only when $T$ rounds have elapsed from the beginning, all good agents complete the \CIST stage and the next two consecutive phases of the \GST stage in round $r^*$.
Since a reliable group has already been created, all good agents complete the previous algorithm before round $r^*$.

Next, we prove that all good agents terminate on $v_{min}$ at the same round.
From the above discussion, all good agents wait on $v_{min}$ in round $r^*$.
Since all good agents obtain the same information on $v_{min}$, they decide the same value on $\GEF$.
Hence, they can terminate at the same round immediately after at least $\GEF+1$ agents execute $flag_t\gets True$.

Lastly, we prove that good agents terminate in at most $3X_N +3(2\lfloor\log\Lambda_{all}\rfloor $ $+f+7)$ $(3X_N +1)+1$ rounds. 
Similarly to Theorem \ref{theorem_NotDetectGathering}, all good agents complete the previous algorithm and gather on $v_{min}$ in at most $T_1=X_N+3(2\lfloor\log\Lambda_{good}\rfloor +f+7)(3X_N +1)$ rounds. %
In addition, since $\IDM$ is an ID of some agent, good agents wait until at most $T_2=2X_N+3(2\lfloor\log \Lambda_{all} \rfloor +6)(3X_N +1)$ rounds have passed.
Note that good agents execute $flag_t\gets True$ if (a) $X_N$ rounds have passed after they complete the previous algorithm and (b) $T$ rounds have passed after the beginning of the algorithm.
Hence, good agents execute $flag_t\gets True$ in at most $T_3=max\{T_1+X_N,T_2\}\leq 2X_N+3(2\lfloor\log\Lambda_{all}\rfloor +f+7)(3X_N +1)$ rounds after they start the algorithm.
Since all good agents start the algorithm within $X_N$ rounds and they terminate after at least $\GEF+1$ agents execute $flag_t\gets True$, they terminate in at most $X_N+T_3+1=3X_N+3(2\lfloor\log\Lambda_{all}\rfloor +f+7)(3X_N +1)+1$ rounds after the first good agent wakes up.
\end{proof}

\section{Conclusion}
In this paper, we have developed two algorithms that achieve the gathering in weakly Byzantine environments.
We proposed two algorithms that reduce the time complexity compared to the existing algorithm by assuming a strong team of agents.
The proposed algorithms operate under the assumption that the upper bound $N$ of the number of nodes is given to agents, and at least $(4f+4)(f+1)$ good agents exist in the network, where $f$ is the number of Byzantine agents.
The first algorithm achieves the gathering with non-simultaneous termination in $O((f+|\Lambda_{good}|)\cdot X(N))$ rounds, where $|\Lambda_{good}|$ is the length of the largest ID among good agents and $X(N)$ is the number of rounds required to explore any network composed of at most $N$ nodes.
The second algorithm achieves the gathering with simultaneous termination in $O((f+|\Lambda_{all}|)\cdot X(N))$ rounds, where $|\Lambda_{all}|$ is the length of the largest ID among agents.

As future work, it would be interesting to study the trade-off between the time complexity and the ratio of good and Byzantine agents.

\end{document}